\documentclass[12pt]{amsart}

\newtheorem{theorem}{Theorem}[section]
\newtheorem{lemma}[theorem]{Lemma}
\newtheorem{corollary}[theorem]{Corollary}

\theoremstyle{definition}

\theoremstyle{remark}

\numberwithin{equation}{section}

\begin{document}

\title{A reduction of 3-SAT problem to Buchberger algorithm.}

\author{Maiia Bakhova}
\address{Mathematics Department\\
Baton Rouge Community College\\
Baton Rouge, Louisiana}
\email{majabakh@gmail.com}

\subjclass{14Q20, 68Q15,13P10}
\date{}

\begin{abstract}
There is a number of known NP class problems, and majority of them have been 
shown to be equivalent to others. In particular now it is clear that 
constructing of a Gr\"{o}bner
basis must be one of equivalent problems, but there was no example. In the
following paper the reduction is constructed. 
\end{abstract}

\maketitle

\section{Introduction}
\label{introduction}

\noindent\textbf{3-SAT problem.} Let  $x_1, x_2, \ldots , x_k$ denote boolean
variables (their value may be True or False).  
A \textit{literal} is either a variable or its negation. An equation in the
propositional calculus is an expression that can be 
constructed using literals and the operations \textbf{and}, denoted by $\wedge$, 
and \textbf{or}, denoted by $\vee$. An example of such equations is 
$(x_1\wedge \neg x_2) \vee  (\neg x_3\wedge x_4)=\text{False}$. We will 
use exponents 1 and $-1$ to indicate if a variable is with negation or not; 
so $x^1_{i}:=x_{i}$ and $x^{-1}_{i}:=\neg x_{i}$. 
Any Boolean equation has an equivalent form as
$$\bigwedge_{i=1}^n 
(x_{i_1}^{\sigma_{i_1}}  \vee x_{i_2}^{\sigma_{i_2}}\vee  \ldots \vee
x_{i_m}^{\sigma_{i_m}})=\text{True}.$$ 
The problem of finding solution for such equation is called the 
\textbf{Boolean satisfiability problem}.
It has been shown that any satisfiability problem can be transformed
in polynomial number of steps (adding variables as needed) into an equation 
$$
\bigwedge_{i=1}^n (x_{i_1}^{\sigma_{i_1}}  \vee x_{i_2}^{\sigma_{i_2}}\vee x_{i_3}^{\sigma_{i_3}})=\text{True}.
$$
The corresponding problem is called the 3-SAT problem.
There are some obvious cases in which we can see immediately if a solution exists or not. For example if the number of clauses is fewer than 8
some solution always exists. From the other hand, if we consider a maximal number of clauses for a one solution, say 
when all variables are true, then the number of clauses which are true is  
$$\frac{2k(2k-2)(2k-4)}{6} -\frac{k(k-1)(k-2)}{6}
=\frac{7k(k-1)(k-2)}{6},$$ 
hence if there are more clauses, then a solution does not exists.

3-SAT problem is one of so called NP problems, see~\cite{sat}. 
There are different versions of required answer for such problem. 
One of them is  to find a solution which satisfies 
the Boolean equation. Another is Yes/No, meaning that one should
indicate if any solution exists.

\noindent\textbf{Gr\"{o}bner basis.} Gr\"{o}bner basis is defined as 
a basis of ideal in multivariate polynomial ring generated by a given set of polynomials,
see~\cite{grbnr}.
The basis is required to have specific properties which help quickly and 
efficiently to resolve a
number of questions about the ideal. For example, it allows to get quick 
answers about its zero locus. In particular, if the zero locus is empty set,
then the corresponding Gr\"{o}bner basis equals 1. The process of finding
a Gr\"{o}bner basis is called \textbf{Buchberger algorithm}.

\section{The Reduction}
\label{reduction}

Consider a usual 3-SAT problem, which is a Boolean equation with $k$ variables
 $x_1, x_2, \ldots , x_k$.  Then the 3-SAT problem
can be transformed into a system of $n$ Boolean equations:
\begin{equation}\label{3sat}
\begin{cases}
x_{1_1}^{\sigma_{1_1}}\vee  x_{1_2}^{\sigma_{1_2}}\vee x_{1_3}^{\sigma_{1_3}} =\text{True} \\
\cdots \cdots \\
x_{n_1}^{\sigma{n_1}}\vee  x_{n_2}^{\sigma_{n_2}}\vee x_{n_3}^{\sigma_{n_3}} =\text{True}
\end{cases}
\end{equation} All variables $x_{1_1}$, \ldots, $x_{n_3}$ (usually with repetitions)  in the 
system are from the set  $\{ x_1, x_2, \ldots , x_k \}$. All exponents can
admit only values 1 or -1.

{\bf Conversion procedure.} We will replace the system of Boolean equations 
with a system 
of real polynomial equations with the following formal process. Introduce new real valued 
variables $z_1, z_2, \ldots , z_k$.  I want each appearance
of a literal $x_{i_j}^{\sigma_{i_j}}$ replace with $(z_{i_j}-1)$ or
$z_{i_j}$ depending on the value of exponent $\sigma_{i_j}$. 
For this purpose we introduce constants $c_{1_1}$, \ldots, $c_{n_3}$, which values are defined by exponents:
\begin{equation}c_{i_j}=
\begin{cases} 
1, \text{ if }  \sigma_{i_j}=1\\
0, \text{ if }  \sigma_{i_j}=-1\\
\end{cases}
\end{equation}
Now let us substitute $x_{i_j}^{\sigma_{i_j}}$ with $(z_{i_j}-c_{i_j})$
 Replace a
boolean operation $\vee$ with a usual multiplication and the "True" value with 0. 
Now the system looks like a polynomial equation system:
\begin{equation}
\begin{cases}
(z_{1_1}-c_{1_1})  (z_{1_2}-c_{1_2})(z_{1_3}-c_{1_3}) =0 \\
\cdots \cdots \\
(z_{n_1}-c_{n_1})  (z_{1_2}-c_{1_2})(z_{1_3}-c_{1_3}) =0
\end{cases}
\end{equation} We say that such system is associated to the system of
Boolean equations \ref{3sat}.

\begin{lemma}
A set of Boolean values $\{ R_1, R_2, \ldots , R_k\}$ 
is a solution for 3-Sat problem iff the number set
$\{ r_1, r_2, \ldots , r_k\}$  is a solution for the associated polynomial
equation system, where $r_i=1$ if $R_i=$True and $r_i=0$ if $R_i=$False for 
all $i=1,2,\cdots, k$.
\end{lemma}

\begin{proof}
Without loss of generality we can assume that the problem~\ref{3sat} includes 
an equation
$x_1^{\sigma_1}\vee  x_2^{\sigma_2}\vee x_3^{\sigma_3} =\text{True}$.
The associated polynomial equation will be 
$(z_1-c_1) (z_2-c_2) (z_3-c_3) =0$.

Let us consider a case when $\sigma_1=1$. Observe that when we assign a 
Boolean variable $x_1$ to be True, and the corresponding variable $z_1$ is~1,
then the clause on the right is True, and the polynomial at the left
side of the equation associated to such clause 
vanishes, because it has a factor $(z_1-1)$. 
In the case when the Boolean variable is False, then the value of the clause
is defined by values of other variables. The corresponding factor in 
polynomial expression 
is non zero and the whole expression can vanish only if there is other factor 
which is zero.
The same is true in reverse when we start with assigned values for 
variables of 
the polynomial system. When the numerical variable $z_1$ is 1, then 
the associated polynomial vanishes. The
corresponding Boolean variable $x_1=$True and the initial Boolean equation is
true.

The case $\sigma_1=-1$ is treated similarly. Now we can generalize the reasoning
for all clause equations.
\end{proof}
\begin{corollary}
A 3-Sat problem has a solution iff the corresponding polynomial equation system has a numeric solution. 
\end{corollary}
\begin{proof}
This is trivial.
\end{proof}
\begin{theorem}
A 3-SAT problem can be resolved using Buchberger algorithm.
\end{theorem}
\begin{proof}
Given a 3-SAT problem construct the associated
system of polynomial equations. Take all polynomials on the left 
hand side of the equations and use them to define an ideal. Now use Buchberger
algorithm to find solutions. The question about the existence of the solution
will be answered, and there are different approaches to find a particular
solution.
\end{proof}

Note that the same reduction can be done for a 2-SAT problem and it yields a 
quadratic equation system. Then choosing a specific value of one of the 
variable turns some equations into linear ones with simple solution. The process
forces other variables to take specific values or show that they can take any value. This is
similar to how the equivalent graph problem is solved. Now the reason why 
the 3-SAT and 2-SAT problems are so different becomes apparent in polynomial
interpretation, which answers a question of Pr. E. M. Luks raised in a private
discussion.

\section*{Acknowledgments}
The author thanks E. M. Luks for introduction to the problem and helpful discussions.

\end{document}